\documentclass[a4paper,runningheads,oneside,11pt]{llncs}

\usepackage{amsmath,amsfonts,a4wide,url}
\usepackage{cite}

\newcommand{\card}[1]{\left| #1 \right|}
\newcommand{\calr}{\mathcal{R}}

\newcommand{\Z}{\mathbb{Z}}
\newcommand{\barZ}{\overline{\Z}}

 \newcommand\etal{\emph{et~al.\@}}

\newcommand{\NP}{\mathrm{NP}}
\newcommand{\WII}{\mathrm{W}[2]}
\newcommand{\PNP}{\ensuremath{\mathrm{P} = \mathrm{NP}}}

\newcommand{\Medit}{\textsc{Edit}}
\newcommand{\mMedit}{\textsc{Min} \Medit}
\newcommand{\mfst}{\textsc{Minimum-Flip Supertree}}
\newcommand{\mfct}{\textsc{Minimum-Flip Consensus Tree}}

\newcommand{\rti}{\textup{RTI}}
\newcommand{\mrti}{\textsc{Min} \rti}

\newcommand{\triplet}[3]{\left\langle  #1, #2 \middle|  #3 \right\rangle}

\newcommand{\sfm}{\mathsf{M}}

\newcommand{\defpb}[4]{
\begin{trivlist}
 \item[]\emph{Name}: #1
 \item[]\emph{Input}: #2
 \item[]\emph{Solution}: #3
 \item[]\emph{Measure}: #4
\end{trivlist}
}

 \newcommand{\defpar}[4]{
 \begin{trivlist}
  \item[]\emph{Name}: #1
  \item[]\emph{Input}: #2
   \item[]\emph{Question}: #4
\item[]\emph{Parameter}: #3
 \end{trivlist}
 }

\begin{document}

\sloppy 

\title{Intractability of the Minimum-Flip Supertree problem and its variants}

\author{
Sebastian B\"ocker\inst{1} \and 
Quang Bao Anh Bui\inst{1} \and 
Fran\c{c}ois Nicolas\inst{1} \and 
Anke Truss\inst{1}}

\institute{
Lehrstuhl f{\"u}r Bioinformatik,
Friedrich-Schiller-Universit{\"a}t Jena, 
Ernst-Abbe-Platz 2, Jena, Germany,
\url{{sebastian.boecker, quangbaoanh.bui, francois.nicolas, anke.truss}@uni-jena.de}
}

\authorrunning{S.~B\"ocker, Q.B.A.~Bui, F.~Nicolas, A.~Truss}

\date{\today}

\maketitle



\begin{abstract}
Computing supertrees is a central problem in phylogenetics.
The supertree method that is by far the most widely used today was introduced in 1992 and is called \emph{Matrix Representation with Parsimony analysis (MRP)}.
\emph{Matrix Representation using Flipping (MRF)}, which was introduced in 2002, is an interesting variant of MRP:
MRF is arguably more relevant that MRP and various efficient implementations of MRF have been presented. 
From a theoretical point of view, implementing MRF or MRP is solving NP-hard optimization problems.
The aim of this paper is to study the approximability and the fixed-parameter tractability of the optimization problem corresponding to MRF, 
namely Minimum-Flip Supertree.
We prove strongly negative results.
\end{abstract}


\section{Introduction}

When studying the evolutionary relatedness of current taxa, 
the discovered relations are usually represented as rooted trees, called \emph{phylogenies}.
Phylogenies for various taxa sets are routinely inferred from various kinds of molecular and morphological data sets.
A subsequent problem is computing \emph{supertrees} \cite{bininda-emonds04phylogenetic}, \emph{i.e.}, 
amalgamating phylogenies for non-identical but overlapping taxon sets to obtain more comprehensive phylogenies. 
Constructing  supertrees is easy if no contradictory information is contained in the data \cite{aho81inferring}.  
However, incompatible input phylogenies are the rule rather than the exception in practice.
The major problem for supertree methods is thus dealing with incompatibilities.

The supertree method that is by far the most widely used today  was independently proposed by Baum \cite{baum92combining} and Ragan \cite{ragan92phylogenetic} in 1992; 
it is called  \emph{Matrix Representation with Parsimony analysis (MRP)} \cite{bininda-emonds04phylogenetic}.
From a theoretical point of view, 
implementing MRP is designing an algorithm for an $\NP$-hard optimization problem~\cite{day86computational-a,foulds82steiner}, 
so the running times of MRP algorithms are sometimes prohibitive for large data sets.

In 2002, 
Chen \etal{} proposed a  variant of MRP \cite{chen02supertrees}, 
which was later called \emph{Matrix Representation using Flipping (MRF)} \cite{chen03flipping}.
MRF is arguably more relevant than MRP \cite{bininda-emonds04phylogenetic} (see also \cite{chen06minimum-flip,eulenstein04performance}), 
and various efficient implementations of MRF have been presented \cite{eulenstein04performance,chen06improved-heuristics,chimani10exact}.
However, as in the case of MRP, 
implementing MRF is designing an algorithm for an $\NP$-hard optimization problem~\cite{chen06minimum-flip}, 
namely \mfst. 
The aim of the present paper is to study the approximability and the fixed-parameter tractability \cite{flum06parameterized} of \textsc{Minimum-Flip Supertree}.
We prove strongly negative results.


\section{Preliminaries}

For each finite set $X$, 
the cardinality of $X$ is denoted $\card{X}$.
The ring of integers is denoted $\Z$.
Define $\barZ = \Z \cup \{ - \infty, + \infty \}$.

\subsection{Rooted phylogenies}

Let $S$ be a finite set.
A \emph{(rooted) phylogeny} for $S$ is a subset $T$ of the power set of $S$ 
that satisfies the following properties:
$\emptyset \in T$,
$S \in T$,
$\{ s \} \in T$ for all $s \in S$, and 
$X \cap Y \in \{ \emptyset, X, Y \}$ for all $X$, $Y \in T$.
The elements of $S$ are the \emph{leaves} of $T$.
The elements of $T$ are the \emph{clusters} of $T$.
 The most natural  representation of $T$ is, of course, a rooted graph-theoretic tree with $\card{T} - 1$ nodes
 (the empty cluster does not correspond to any vertex).

Given two phylogenies $T_1$ and $T_2$ for $S$,
$T_1$ is a subset of $T_2$ 
if, and only if, 
the graph representation of $T_1$ 
can be obtained from 
the graph representation of $T_2$ by contracting (internal) edges.
If $T_1$ is a subset of $T_2$ and if we assume that hard polytomies never occur then $T_2$ is at least as informative as $T_1$.

\subsection{Bipartite graphs and perfect phylogenies}

A \emph{bipartite graph} is a triple $G = (C, S, E)$, where 
$C$ and $S$ are two finite sets 
and 
$E$ is a subset of $C \times S$.
The elements of $E$ are the \emph{edges} of $G$.
The elements of $(C \times S) \setminus E$ are the \emph{non-edges} of $G$. 
For each $c \in C$, $N_G(c)$ denotes the \emph{neighborhood} of $c$ in $G$:  
$N_G(c) = \left\{ s \in S : (c, s) \in E \right\}$.

Let $\sfm(G)$ denote the set of all quintuples $(s, c, s', c', s'') \in S \times C \times S \times C \times S$ such that 
$(c, s) \in E$, 
$(c, s') \in E$, 
$(c', s') \in E$, 
$(c', s'') \in E$, 
$(c, s'') \notin E$, and
$(c', s) \notin E$. 
The latter conditions state that the bipartite graph depicted in  \cite[Figure~4]{peer04incomplete} is an induced subgraph of $G$.
A \emph{perfect phylogeny} for $G$ is a phylogeny $T$ for $S$ such that $N_G(c)$ is a cluster of $T$ for every $c \in C$.
We say that $G$ is \emph{$\sfm$-free} \cite{chen02supertrees,chen06minimum-flip,chimani10exact,boecker08improved,komusiewicz08cubic}
(or \emph{$\mathsf{\Sigma}$-free}) \cite{peer04incomplete,bininda-emonds04phylogenetic,chen03flipping}) 
if the following three equivalent conditions are met: 
\begin{enumerate}
 \item for all $c$, $c' \in C$, $N_G(c) \cap N_G(c') \in \left\{ \emptyset,  N_G(c), N_G(c') \right\}$,
 \item $\sfm(G)$ is empty, and  
 \item there is a perfect phylogeny for $G$.
\end{enumerate}
Put $T_G = \{ \emptyset, S \} \cup \left\{ N_G(c) : c \in C  \right\} \cup \left\{ \{ s \} : s \in S \right\}$.
If $G$ is $\sfm$-free then $T_G$ satisfies the following two properties:
\begin{enumerate} 
 \item $T_G$ is a perfect phylogeny for $G$. 
 \item $T_G$ is a subset of any perfect phylogeny for $G$.
\end{enumerate}

\paragraph{Modelization.}
In our model,
$S$ is a set of \emph{species} (or more generally \emph{taxa})
and
$C$ is a set of binary \emph{characters}.
For each $(c, s) \in C \times S$,  
$(c, s) \in E$ means that species $s$ possesses character $c$ and
$(c, s) \notin E$ means that species  $s$ does not possess character $c$.
Character data come from the morphological and/or molecular properties of the taxa \cite{gusfield97algorithms}.
The assumption of the model is that for all $c \in C$ and all $s$, $s' \in S$,
the following two assertions are equivalent:
 \begin{enumerate}
  \item Both species $s$ and $s'$ possess character $c$.
  \item Some common ancestor of species $s$ and $s'$ possesses character $c$.
 \end{enumerate}
A phylogeny for $S$ satisfies the assumption of the model if, and only if, it is a perfect phylogeny for $G$.

\subsection{Dealing with incomplete and/or erroneous data sets}

A \emph{bipartite draft-graph} (or \emph{weighted bipartite fuzzy graph} \cite{bodlaender10clustering}) is a triple $H = (C, S, F)$, where 
$C$ and $S$ are two finite sets 
and 
$F$ is a function from $C \times S$ to $\barZ$.
The function $F$ is the \emph{weight function} of $H$.
The range of $F$ is called the \emph{weight range} of $H$.
An \emph{edge} of $H$ is an element $e \in C \times S$ such that  $F(e) \ge 1$.
A \emph{joker-edge} of $H$ is an element $e \in C \times S$ such that  $F(e) = 0$.
A \emph{non-edge} of $H$ is an element $e \in C \times S$ such that  $F(e) \le -1$.

For each $e \in C \times S$, the magnitude of $F(e)$ is the \emph{edit cost} of $e$ in $H$.
An \emph{edition} of $H$ is a bipartite graph $G$ of the form $G = (C, S, E)$ for some subset $E \subseteq C \times S$. 
A \emph{conflict} between $G$ and $H$ is an element  $e \in C \times S$ that satisfies one of the following two conditions:
 \begin{enumerate}
  \item $e$ is an edge of $G$ and $e$ is a non-edge of $H$ or 
  \item $e$ is a non-edge of $G$ and $e$ is an edge of $H$. 
 \end{enumerate}
The sum of the edit costs in $H$ over all conflicts between $G$ and $H$ is denoted $\Delta(G, H)$:
$$
\Delta(G, H) = \sum_{e \in E} \max \{ 0, - F(e) \} + \sum_{e \in E \setminus (C \times S)}   \max \{ 0,  F(e) \} \, .
$$

The following minimization problem and its (parameterized) decision version generalize several previously studied problems:

\defpb{\textsc{Minimum} $\sfm$-\textsc{Free Edition} or \mMedit.}
{A bipartite draft-graph $H$.}
{An $\sfm$-free edition $G$ of $H$.}
{$\Delta(G, H)$.}

\defpar{$\sfm$-\textsc{free Edition} or \Medit.}
{A bipartite draft-graph $H$ and an integer $k \ge 0$.}
{$k$.}
{Is there an $\mathsf{M}$-free edition $G$ of $H$ such that $\Delta(G, H) \le k$?}

For each subset $X \subseteq  \barZ$,
define \mMedit-$X$ as the restriction of \mMedit{} to those bipartite draft-graphs whose weight ranges are subsets of $X$,
and similarly, 
define \Medit-$X$ as the restriction of \Medit{} to  those instances $(H, k)$ such that the weight range of $H$ is a subset of $X$.
Notably, 
\mMedit-$\{ - 1, + 1 \}$ is the \mfst{} problem 
and its restiction \mMedit-$\{ - 1, 0, + 1 \}$ is the \mfct{} problem
 \cite{boecker08improved,komusiewicz08cubic,chen06minimum-flip,chen02supertrees,chen03flipping,chimani10exact,chen06improved-heuristics, eulenstein04performance, bininda-emonds04phylogenetic}.

\paragraph{Modelization.}
Incomplete and/or possibly erroneous character data sets are naturally modeled by bipartite draft-graphs:
joker-edges represent incompletenesses and edit costs allow parsimonious error-corrections.

\paragraph{Supertrees.}
The most interesting feature of \mMedit{} is that it can be thought as a supertree construction problem, 
and more precisely, 
the optimization problem underlying MRF
\cite{bininda-emonds04phylogenetic, chen06minimum-flip, chen03flipping, chen02supertrees}.

\section{Previous results}

\mMedit-$X$ has been studied for several subsets $X \subseteq \barZ$ \cite{gusfield91efficient, gusfield97algorithms, peer04incomplete, boecker08improved, komusiewicz08cubic, chen06minimum-flip, chen03flipping, chimani10exact, chen06improved-heuristics, eulenstein04performance,bininda-emonds04phylogenetic}, sometimes implicitely.
Let $H = (C, S, F)$ be a bipartite draft-graph and let $k$ be a non-negative integer.

Put 
 $\barZ_+ = \left\{ n \in \barZ : n \ge 0 \right\}$ 
and  
$\barZ_- = \left\{ n \in \barZ : n \le 0 \right\}$.
If  $H$ has no non-edge, or equivalently, if the weight range of $H$ is a subset of $\barZ_+$ then 
the complete bipartite graph $K = (C, S, C \times S)$  is an $\sfm$-free edition of $H$ such that $\Delta(K, H) = 0$.
In the same way, 
if $H$ has no edge then 
the empty bipartite graph $\overline{K} = (C, S, \emptyset)$ is an $\sfm$-free edition of $H$ such that $\Delta(\overline{K}, H) = 0$.
Hence, 
\mMedit-$\barZ_+$ and \mMedit-$\barZ_-$  are trivial problems.

Now, consider the case where the weight range of $H$ is a subset of $\left\{ - \infty,  + \infty \right\}$. 
The bipartite graph 
$$G = \left(C, S, \left\{ e \in C \times S : F(e) = + \infty  \right\} \right)$$
 is an edition of $H$ such that 
$\Delta(G, H) = 0$;
for every edition $G'$ of $H$, $G' \ne G$ implies $\Delta(G', H) = +\infty$.
Therefore, 
solving \mMedit{} on $H$ reduces to deciding whether $G$ is $\sfm$-free, 
which can be achieved in $O( \card{C} \card{S} )$ time \cite{gusfield91efficient,gusfield97algorithms}.
Hence, 
\mMedit-$\left\{ - \infty,  + \infty \right\}$ can be solved in polynomial time because it reduces to the recognition problem associated with the class of $\sfm$-free bipartite graphs.
More generally, \mMedit-$\left\{ - \infty, 0, + \infty \right\}$ can also be solved in polynomial time 
because it reduces to the sandwich problem \cite{golumbic95sandwich} associated with the class of $\sfm$-free bipartite graphs: 
in the case where the the weight range of $H$ is a subset of $\left\{ - \infty, 0, + \infty \right\}$, 
\mMedit{}  can be solved on $H$ in $\widetilde O( \card{C} \card{S} )$ time \cite{peer04incomplete}. 

Put 
$I = \left\{ -1,  +\infty \right\}$,
$D = \left\{ -\infty, + 1 \right\}$, and
$U = \left\{ - 1, + 1 \right\}$. 
\mMedit-$I$, 
\mMedit-$D$, 
and 
\mMedit-$U$ (also known as \mfct)
are the three unweighted edge-modification problems \cite{natanzon01edge} associated with the class of $\sfm$-free bipartite graphs:
\mMedit-$I$ is the insertion (or completion) problem
and 
\mMedit-$D$ is the deletion problem.
\Medit-$I$, 
\Medit-$D$, 
and 
\Medit-$U$
are $\NP$-complete \cite{chen06minimum-flip}.

Put $\barZ^* = \barZ \setminus \{ 0 \}$. 
\mMedit-$\barZ^*$ is the restriction of \mMedit{} to those bipartite draft-graphs that have no joker-edge.
The most positive result concerning \mMedit{} is that \Medit-$\barZ^*$ is FPT:
in the case where the weight range of $H$ is a subset of $\barZ^*$,
deciding whether $(H, k)$ is a yes-instance of \Medit{} 
(and if so, computing an $\sfm$-free edition $G$ of $H$ such that $\Delta(G, H) \le k$) 
can be achieved in $O(6^k \card{C} \card{S})$ time \cite{chen06minimum-flip}.
Better FPT algorithms have been presented for the special cases 
\mMedit-$I$ \cite{chen06minimum-flip}, 
\mMedit-$D$ \cite{chen06minimum-flip}, and \mMedit-$U$ \cite{boecker08improved,komusiewicz08cubic}.
In particular, \Medit-$U$ has a polynomial kernel  \cite{komusiewicz08cubic}.

Exact algorithms based Integer Linear Programming \cite{chimani10exact}, 
as well as heuristics \cite{eulenstein04performance,chen06improved-heuristics},
have been tested for \mMedit-$\left\{ - 1, 0, + 1 \right\}$  (also known as \mfst).

\section{Contribution} 

The aim of the present paper is to complete the study of \mMedit{} by proving:

\begin{theorem}  \label{th:edit}
For all $\alpha$, $\beta \in \barZ$ such that $- \alpha < 0 < \beta$ and $(\alpha, \beta) \ne (+ \infty, + \infty)$,
the following two statements hold:
\begin{enumerate}
 \item  \Medit-$\{ - \alpha , 0, \beta \}$ is $\WII$-hard and 
 \item if there exists a real constant $\rho \ge 1$ such that  \mMedit-$\{ - \alpha , 0, \beta \}$ is $\rho$-approximable in polynomial time then $\PNP$.
\end{enumerate}
\end{theorem}
The intractabilities of 
\mMedit-$\left\{ - 1, 0, + \infty \right\}$,
\mMedit-$\left\{ - \infty, 0, + 1 \right\}$, and 
 \mMedit-$\left\{ - 1, 0, + 1 \right\}$  (also known as \mfst)
follow from Theorem~\ref{th:edit}.

Our proof of Theorem~\ref{th:edit} requires the introduction of some material and results from the literature \cite{byrka2010new}.
For all $x$, $y$, $z$, 
let 
$\triplet{x}{y}{z}$ denote the unique phylogeny for $\{ x, y, z \}$ having $\{ x, y \}$ as a cluster:
$$
\triplet{x}{y}{z} = \left\{ \emptyset, \{ x \}, \{ y \}, \{ z \}, \{ x, y \},  \{ x, y, z \} \right\} \, .
$$
A \emph{resolved triplet} is a phylogeny of the form $\triplet{x}{y}{z}$ for some pairwise distinct $x$, $y$, $z$. 
Given a phylogeny $T$ for some superset of $\{ x, y, z \}$, 
we say that $\triplet{x}{y}{z}$ \emph{fits} $T$ 
if there exists a cluster $X$ of $T$ such that $X \cap  \{ x, y, z \} = \{ x, y \}$.

\defpb{\textsc{Minimum Resolved Triplets Inconsistency} or \mrti.}
{A finite set $S$ and a set $\calr$ of resolved triplets with leaves in $S$.}
{A phylogeny $T$ for $S$.}
{The number of those elements of $\calr$ that do not fit $T$.}  

\defpar{\textsc{Resolved Triplets Inconsistency} or \rti.}
{A finite set $S$, a set $\calr$ of resolved triplets with leaves in $S$, and an integer $k \ge 0$.}
{$k$.}
{Is there a phylogeny $T$ for $S$ such that at most $k$ elements of $\calr$ do not fit  $T$?}

\begin{theorem}[Byrka, Guillemot, and Jansson 2010 \cite{byrka2010new}]  \label{th:rti}\begin{enumerate}
\item \rti{} is $\WII$-hard. 
\item 
If there exists a real constant $\rho \ge 1$ such that  \mrti{} is $\rho$-approximable in polynomial time then $\PNP$.
\end{enumerate}
\end{theorem}

The idea behind the proof of Theorem~\ref{th:edit} is the following:
given an instance $(S, \calr)$ of \mrti, 
computing a ``good'' solution of \mrti{} on $(S, \calr)$ is computing a ``good'' MRF  supertree for the phylogenies in $\calr$. 

\begin{proof}[Proof of Theorem~\ref{th:edit}.1]
Theorem~\ref{th:edit}.1 is deduced from Theorem~\ref{th:rti}.1: 
we show that \rti{} FPT-reduces to \Medit-$\{ -\alpha, 0, \beta \}$.
Put $\gamma = \min \{ \alpha, \beta \}$. 
Note that $\gamma$ is a positive integer.

Let $(S, \calr, k)$ be an arbitrary instance of \rti.
The reduction maps $(S, \calr, k)$  to an instance $(H, \gamma k)$ of \Medit-$\{ -\alpha, 0, \beta \}$, where $H$ is as follows.
Let  $C = \left\{ 1, 2, \dotsc,  \card{\calr} \right\}$.
Write $\calr$ in the form 
$$\calr = \left\{ \triplet{x_c}{y_c}{z_c} : c \in C \right\} \, .$$
Let $F$ be the function from $C \times S$ to $\barZ$ given by:
$$
F(c, s) = 
\begin{cases}
\beta & \text{if $s \in \{ x_c, y_c \}$} \\
  0 & \text{if $s \notin \{ x_c, y_c, z_c \}$} \\
- \alpha & \text{if $s = z_c$} 
\end{cases}
$$
for all $(c, s) \in C \times S$.
Let $H = (C, S, F)$.

Clearly  $(H, \gamma  k)$ is computable from $(S, \calr, k)$ in polynomial time.
It remains to prove that 
$(S, \calr, k)$ is a yes-instance of  \rti{} 
if, and only if, 
 $(H, \gamma k)$ is a yes-instance of \Medit.

\paragraph{If.}
Assume that $(H, \gamma k)$ is a yes-instance of \Medit{}.
Then, there exists an $\sfm$-free edition $G$ of $H$ such that $\Delta(G, H) \le \gamma k$.
Let $C'$ denote the set of all $c \in  C$ such that $(c, s)$ is a conflict between $G$ and $H$ for at least one $s \in \{ x_c, y_c, z_c \}$.
Since there are at least $\card{C'}$ conflicts between $G$ and $H$, 
we have $\gamma \card{C'} \le \Delta(G, H)$,
and thus $\card{C'} \le k$.
Let $T$ be a perfect phylogeny for $G$.
For each $c \in C \setminus C'$, 
we have $N_G(c) \cap \{ x_c, y_c, z_c\} = \{ x_c, y_c \}$, 
and thus  $\triplet{x_c}{y_c}{z_c}$ fits  $T$. 
Hence, $T$ is a phylogeny for $S$ such that at most  $k$ elements of $\calr$ do not fit $T$. 
Therefore, $(S, \calr, k)$ is a yes-instance of \rti.

\paragraph{Only if.}
Assume that $(S, \calr, k)$ is a yes-instance of \rti. 
 Then, there exists a phylogeny $T$ for $S$ such that at most  $k$ elements of $\calr$ do not fit $T$. 
Let $C'$ denote the set of all $c \in C$ such that $\triplet{x_c}{y_c}{z_c}$ does not fit $T$.
For each $c \in C \setminus C'$, 
let $X_c$ be a cluster of $T$ such that $X_c \cap \{ x_c, y_c, z_c \} = \{ x_c, y_c \}$.
If  $\alpha \le \beta$ then let $X_c = S$ for each $c \in C'$;
if $\beta < \alpha$ then let $X_c = \{ x_c \}$ for each $c \in C'$.
Put $G =  \left( C, S, \bigcup_{c \in C} \{ c \} \times X_c   \right)$.
\begin{enumerate}
 \item $G$ is an edition of $H$.
 \item $T$ is a perfect phylogeny for $G$ because $N_G(c) = X_c$ is a cluster of $T$ for all $c \in C$.
Therefore, $G$ is $\sfm$-free.
 \item Let $\Gamma$ denote the set of all conflicts between $G$ and $H$.
If $\alpha \le \beta$ then  $\Gamma = \left\{ (c, z_c) : c \in C' \right\}$; 
if $\beta < \alpha$ then $\Gamma = \left\{ (c, y_c) : c \in C' \right\}$.
The edit cost in $H$ of every conflict between $G$ and $H$ equals $\gamma$.
Therefore, we have $\Delta(G, H) = \gamma \card{ \Gamma } = \gamma \card{C'} \le \gamma k$. 
\end{enumerate}
Hence, $(H, \gamma k)$ is a yes-instance of \Medit.
\qed
\end{proof}

\begin{proof}[Proof of Theorem~\ref{th:edit}.2.]
Let  $\rho$ be real number greater than or equal to $1$.
It follows from the proof of Theorem~\ref{th:edit}.1 that
if $G$ is a $\rho$-approximate solution of \mMedit{} on $H$ 
then any perfect phylogeny for $G$ is a $\rho$-approximate solution of \mrti{} on $(S, \calr)$.
Therefore, if \mMedit{} is $\rho$-approximable in polynomial time then \mrti{} is also $\rho$-approximable in polynomial time.
It is now clear that Theorem~\ref{th:edit}.2 follows from Theorem~\ref{th:rti}.2.
\qed
\end{proof}

\section{Conclusion}

 To conclude, let us contrast Theorem~\ref{th:edit} with two recent results.

The \textsc{Maximum Parsimony} (MP) problem \cite{alon10approximate} is the $\NP$-hard optimization problem \cite{day86computational-a,foulds82steiner} 
underlying MRP, 
as \mMedit{}  is the optimization problem underlying MRF.
Although \mMedit{} is $\NP$-hard to approximate within any constant factor by Theorem~\ref{th:edit}.2, 
MP is $1.55$-approximable in polynomial time \cite{alon10approximate}.

The parameterized problems \Medit{} and \textsc{Weighted Fuzzy Cluster Editing (WFCE)}  \cite{bodlaender10clustering} are closely related: 
WFCE is the draft-graph edition problem corresponding to the class of $P_3$-free graphs.
\Medit{} is $\WII$-hard by Theorem~\ref{th:edit}.1 but WFCE has been recently shown to be fixed-parameter tractable \cite{marx11fixed-parameter} (see also \cite{bousquet11multicut,demaine06correlation}).

\bibliographystyle{abbrv}
\bibliography{bibtex/group-literature}

\end{document}